\documentclass[runningheads]{llncs}
\usepackage{graphicx}
\usepackage{
  amsmath,
  amsfonts,
  amscd,
  amssymb,
  color,
  enumerate, 
  dirtytalk,
  xcolor,
  etoolbox,
  graphicx,
  comment,      
  mathtools,    
  enumitem,
  mathdots,     
  xifthen,
  bbold,
}

\newtoggle{comments}
\newtoggle{details}
\newtoggle{detailsnote}

\toggletrue{comments}   
\toggletrue{detailsnote}   

\iftoggle{comments}{%
  \newcommand{\jcomments}[1]{
    \ \\
    {\color{purple}
      \textbf{JS:} #1
    }
    \\
    }
}{%
  \newcommand{\jcomments}[1]{}
}

\iftoggle{details}{%
  \newcommand{\details}[1]{
      \ \\
      {\color{olive}
        \textbf{Details:} #1
      }
      \\
  }
}{%
  \newcommand{\details}[1]{}
}

\newcommand\N{\mathbb{N}}

\begin{document}
\title{Zero-Knowledge MIPs using\\ Homomorphic Commitment Schemes}
\titlerunning{Zero-Knowledge MIPs using Homomorphic Commitment Schemes}
%
%
%
\author{Claude Crépeau \and John Stuart}
\institute{McGill University}
\maketitle              
\begin{abstract}
A Zero-Knowledge Protocol (ZKP) allows one party to convince another party of a fact without disclosing any extra knowledge except the validity of the fact. For example, it could be used to allow a customer to prove their identity to a potentially malicious bank machine without giving away private information such as a personal identification number. This way, any knowledge gained by a malicious bank machine during an interaction cannot be used later to compromise the client’s banking account. An important tool in many ZKPs is bit commitment, which is essentially a digital way for a sender to put a message in a lockbox, lock it, and send it to the receiver. Later, the key is sent for the receiver to open the lockbox and read the message. This way, the message is hidden from the receiver until they receive the key, and the sender is unable to change their mind after sending the lockbox. In this paper, the homomorphic properties of a particular multi-party commitment scheme are exploited to allow the receiver to perform operations on commitments, resulting in polynomial time ZKPs for two NP-Complete problems: the Subset Sum Problem and 3SAT. These ZKPs are secure with no computational restrictions on the provers, even with shared quantum entanglement. In terms of efficiency, the Subset Sum ZKP is competitive with other practical quantum-secure ZKPs in the literature, with less rounds required, and fewer computations.

\keywords{relativistic cryptography  \and zero-knowledge protocols \and quantum security.}
\end{abstract}
\section{Introduction}

Zero-knowledge proofs were first introduced by \cite{10.1145/22145.22178}. A couple years later, Blum presented an elegant ZKP for the Hamiltonian cycle problem in \cite{Blum87howto}. A Hamiltonian cycle is a cycle in a graph which passes through each vertex exactly once. After a graph $G$ is fixed, Blum's protocol allows a prover to convince a verifier that $G$ contains a Hamiltonian cycle without giving away any knowledge of the cycle. Like many other ZKPs, Blum's protocol uses bit commitment, which is essentially a secure way for Alice to send Bob a message in a locked box so that he cannot see the message until she sends the key at a later time, but she also cannot change the message after she sends the box. Bit commitment schemes typically rely on computational assumptions such as existence of one-way functions as in \cite{10.1145/116825.116852}, however there is a protocol that was proposed in \cite{https://doi.org/10.48550/arxiv.quant-ph/9806031} that only uses spatial separation between parties as its sole assumption to prove security. In \cite{Chailloux_2017}, the authors adapted Blum's ZKP from \cite{Blum87howto} to use the bit commitment protocol from \cite{https://doi.org/10.48550/arxiv.quant-ph/9806031}, and this new protocol was proven to be secure against quantum adversaries. 

One interesting aspect of the bit commitment scheme in \cite{https://doi.org/10.48550/arxiv.quant-ph/9806031} is its homomorphic properties. In particular, the contents of two commitments can be added together before unveiling, and in general, it is possible to form a new commitment containing any linear combination of the contents of the commitments. For this reason, along with rising public interest in homomorphic encryption, it would be an interesting result to have a ZKP that explicitly uses this homomorphic property while still being secure against quantum adversaries. Similar ideas have been considered for classical adversaries in \cite{4568210} and \cite{10.1007/3-540-47721-7_16}, where a perfectly hiding and computationally binding bit commitment scheme is used to create ZKPs that rely on the commitment scheme's homomorphic properties for many steps.

In addition to its homomorphic properties, the commitment scheme of \cite{https://doi.org/10.48550/arxiv.quant-ph/9806031} can be used to commit to values other than bits as long as the field $\mathbb{F}_Q$ is chosen carefully. An excellent NP-Complete problem to showcase this homomorphic property of the bit commitment scheme is the Subset Sum Problem, which asks the following: Given a set $S$ of positive integers and a target integer $k$, is there a subset of $S$ which sums to $k$ \cite{10.5555/1051910}? In Section \ref{sec:SSP}, a novel ZKP is given for this problem, and it is shown to be efficient. In particular, when compared to other practical quantum-secure ZKPs in the literature, our protocol requires less rounds and fewer computations than other protocols \cite{Alikhani_2021},
 \cite{Chailloux_2017}, \cite{https://doi.org/10.48550/arxiv.2112.01386}. Then in Section \ref{sec:3SAT}, another new ZKP is given for the well-known 3SAT problem which is also NP-Complete \cite{AhoAlfredV1974Tdaa}. It exploits a technique where the prover can convince the verifier that two committed bits are equal, or opposite, without unveiling them. Both ZKPs in this paper use a technique similar to that used in \cite{Chailloux_2017} in order to prove soundness against quantum adversaries. 

\section{Relativistic Bit Commitment}\label{sec:REL-CS}

In \cite{https://doi.org/10.48550/arxiv.quant-ph/9806031} the authors introduce a bit commitment scheme involving two provers and two verifiers. Several variations of the protocol have had their security examined in \cite{jofc-2005-14207}, \cite{10.1007/978-3-642-25385-0_22}, and \cite{Lunghi_2015}. Here, we will present the scheme, a brief explanation to why it is hiding and binding, then finally explore some of its useful properties. 

\subsection{Commitment Protocol}

The scheme requires two separated provers P1 and P2, as well as two verifiers V1 and V2. The operations will be over the field $\mathbb{F}_Q$ for some large prime power $Q$. Suppose that the provers want to commit to an element $b\in\mathbb{F}_Q$. They choose a random $c\in\mathbb{F}_Q$, and the protocol goes as follows:\\
~\\
~\\

\textit{Commit Phase}
\begin{enumerate}
    \item V1 chooses a random $a\in \mathbb{F}_Q$ and sends $a$ to P1.
    \item P1 replies with $w:=a b+c$ to V1. 
\end{enumerate}

Now the provers have committed to $b$. When the provers would like the verifiers to know their hidden value $b$, then they complete the unveil phase:\\

\textit{Unveil Phase}
\begin{enumerate}
    \item P2 sends $b,c$ to V2.
    \item V1 and V2 verify that $w=a b+c$.
\end{enumerate}

This protocol is illustrated in Figure \ref{fig:bc}.
\begin{figure}[ht!]
    \centering
    \includegraphics[scale=0.25]{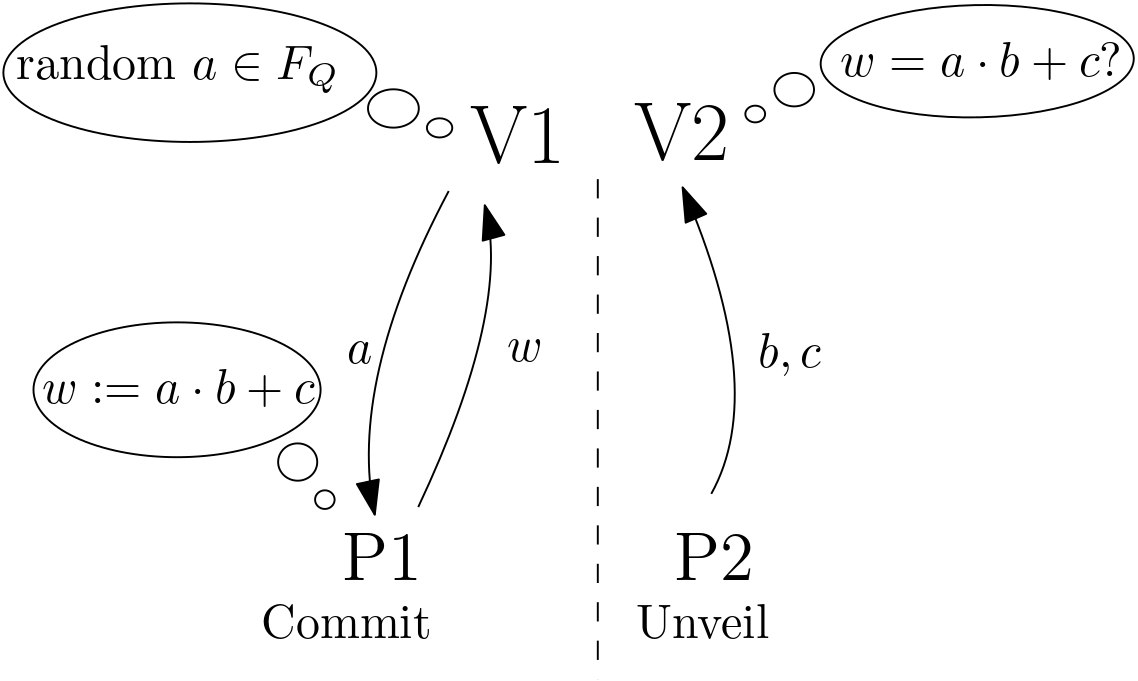}
    \caption{Relativistic Two-Prover Bit Commitment Scheme\label{fig:bc}}
\end{figure}

\subsection{Hiding Property}

We will now explain why the protocol is perfectly hiding. Once V1 has sent $a$ and received $w$, then for each possible $b$ that P1 could be committing to, there is a unique $c$ such that $c=w-ab$. Since $b$ and $c$ are inaccessible to V1 and V2 before the unveil phase, this means that no options can be ruled out by the verifiers. In other words, every possible value $b$ is equally likely from the verifiers' point of view.

\subsection{Binding Property}

Now we discuss why the protocol is binding. Put simply, after P1 completes the Commit phase, what keeps P2 from unveiling another value? The formal proof of the sum-binding property against quantum parties is given in \cite{Chailloux_2017}, therefore we will merely give an intuition. Suppose P1 sent $w$ to V1. Then if P2 would like to unveil to $b$, then P2 would have to send $b,c$ such that $w=a\cdot b+c$, whereas if P2 would like to instead unveil to $b'$, then P2 would have to send $b',c'$ such that $w=a\cdot b'+c'$. If P2 knows values $c,c'$ to successfully unveil $b$ or $b'$ at will, then P2 could compute $a$ since the equations above imply that
\begin{align*}
    \frac{c'-c}{b-b'}=a.
\end{align*}

However, P2 should only be able to compute $a$ with probability $\frac{1}{Q}$ since it was chosen randomly by V1 and only sent to P1 who is separated from P2. For this reason, the provers can only change their commitment with negligible probability. As mentioned, the quantum case is analyzed in \cite{Chailloux_2017}.

\subsection{Homomorphic Property}

One special feature of this commitment scheme is that the contents of two commitments can be combined and unveiled. We take advantage of this in order to create the two ZKPs in this paper. 

Imagine P1 commits to $b$ and $b'$, but P2 wants to have the option to unveil the sum $b+b'$ and nothing else. Then by combining the commitments, the verifiers are able to transform the commitments of $b$ and $b'$ into a commitment of the sum $b+b'$. More specifically, after the Commit phase, the verifiers possess the values $w$ and $w'$ such that
\begin{align*}
    w=a\cdot b+c\text{ and } w'=a'\cdot b'+c'.
\end{align*}

Then as long as $a=a'$, the verifiers can add the values they receive to obtain a commitment $w+w'$ of $b+b'$ since
\begin{align*}
    w+w'=(a\cdot b+c)+(a\cdot b'+c')=a(b+b')+(c+c').
\end{align*}

From the expression above, one can see that the key needed to unveil $b+b'$ is $c+c'$. This property is exploited in Section \ref{sec:SSP} where it is the primary mechanism for a simple ZKP for the Subset Sum Problem. Note that this behaviour can be generalized to any linear combination of commitments. It can be useful to unveil the difference $b-b'$ since it is 0 if and only if the commitments $b$ and $b'$ were equal.

\section{Tools}

The main challenge is proving soundness of the ZKPs in this paper. The approach is to translate the analysis into the language of 2-player entangled games and use tools from \cite{Chailloux_2017} to prove soundness. We start with the definition of a game.

\begin{definition}\label{defin:game}
    A game $G=(I_A,I_B,O_A,O_B,V)$ is defined by
    \begin{itemize}
        \item 2 input sets $I_A$,$I_B$ which are respectively Alice’s and Bob’s input sets.
        \item 2 output sets sets $O_A,O_B$ which are respectively Alice’s and Bob’s output sets.
        \item A valuation function $V : I_A \times I_B \times O_A \times O_B \to \{0, 1\}$ which indicates whether the game is won for some fixed inputs and outputs. The game is won if the value of $V$ is 1.
    \end{itemize}
\end{definition}

Next we define the game $G_{coup}$ from $G$.
\begin{definition}\label{defin:gamecoup}
    For any game $G = (I_A,I_B,O_A,O_B,V)$ on the uniform distribution we define $G_{coup}$ as follows:
    \begin{itemize}
        \item  Alice receives a random $x \in I_A$. Bob receives a random pair of different inputs $(y,y')$ from $I_B$. 
        \item Alice outputs $a \in O_A$. Bob outputs $b,b' \in O_B$.
        \item They win the game if $V(x,y,a,b)=V(x,y',a,b')=1.$
    \end{itemize}
\end{definition}

The twist with $G_{coup}$ is that the second prover must answer two questions at once, making it more difficult for Alice and Bob to win $G_{coup}$. In fact, the maximum probabilities of winning the two games are related. If we denote the maximum probability of winning a game $G$ among all quantum strategies by $\omega^*(G)$, then we have the following result, proved in \cite{Chailloux_2017}.

\begin{proposition}\label{prop:bound1}
    For any game $G$ on the uniform distribution which is $S$-projective, we have $$\omega^*(G_{coup})\geq \frac{1}{64S}\left( \omega^*(G)-\frac{1}{|I_B|} \right)^3.$$
\end{proposition}

Note that $S$-projective refers to the definition below:

\begin{definition}\label{defin:proj}
    Let $S\in\N$. We say that a game is $S$-projective if Bob has at most $S$ possible outputs to win the game for any possible $x,y,a$. So 
    \begin{align*}
        \max_{x,y,a} |\{b \mid V(x, y,a, b) = 1\}| \leq S.
    \end{align*}
\end{definition}
We give our own analogous classical result to aid the reader in understanding Proposition \ref{prop:bound1}. 

\begin{proposition}\label{prop:bound}
    For any game $G$ with questions asked uniformly at random and $I_B=\{0,1\}$, the classical winning probabilities have the following relationship: $2 \omega(G)-1\leq \omega(G_{coup})$.
\end{proposition}

We will present our own proof of this result.

\begin{proof}
    Suppose Alice and Bob can win $G$ with probability $\omega(G)$. Then we will design a strategy for $G_{coup}$ that wins with probability at least $2 \omega(G)-1$. On input $(x,(y,y'))$, Alice and Bob do the following:
    \begin{itemize}
        \item Alice runs on $x$ and outputs $a$
        \item Bob runs on $y$ and outputs $b$
        \item Bob is re-winded
        \item Bob runs on $y'$ and outputs $b'$.
    \end{itemize}
    
    Note that since Bob is entirely classical, he can be re-winded without any issues. Now, we analyze the probability that both answers are correct. In other words, what is the probability that $V(x,y,a,b)=1$ and $V(x,y',a,b')=1$? 
    
    Letting $n:=|I_A|$, we denote $v,w\in [0,1]^n$ to be the probability vectors such that 
    \begin{align*}
        v_x=\mathop{\mathbb{E}}[V(x,0,A(x),B(0))]\quad \text{and}\quad w_x=\mathop{\mathbb{E}}[V(x,1,A(x),B(1))]
    \end{align*} 
    where $A(x),B(y)$ are Alice and Bob's outputs on input $x$ and $y$. The expectation is taken over Alice and Bob's random coin flips. This allows us to write the probability of winning $G$ as 
    \begin{align*}
        \omega(G)=\mathop{\mathbb{E}}_{x,y}[V(x,y,A(x),B(y))]=\frac{v\cdot \vec{1}+w\cdot \vec{1}}{2n}.
    \end{align*}
    
    Then since winning $G_{coup}$ requires succeeding for both inputs, we get 
    \begin{align*}
        \omega(G_{coup})\geq\mathop{\mathbb{E}}_{x}[V(x,0,A(x),B(0))V(x,1,A(x),B(1))]= \frac{v\cdot w}{n}.
    \end{align*} 
    
    Then by noting that each entry of $v$ and $w$ is between 0 and 1, we can obtain the result as follows:
    \begin{align*}
        &0 \leq (\vec{1}-v)\cdot (\vec{1}-w) =v\cdot w-v\cdot \vec{1}-w\cdot \vec{1} +\vec{1}\cdot \vec{1} =v\cdot w-2n\omega(G) +n\\
        \implies & 2\omega(G)-1 \leq \frac{v\cdot w}{n} \leq \omega(G_{coup}).
    \end{align*}
    \qed
\end{proof}

The above proposition for the classical case relies on rewinding Bob. In the quantum case, rewinding cannot always be done after measuring the output $b$ since this would allow for copying of quantum states. In \cite{Chailloux_2017}, a proof is given for a theorem that bounds the amount of error that arises from making two consecutive measurements on a quantum state. This theorem is applied to a strategy for $G_{coup}$ that involves running Bob on both inputs consecutively with no rewinding in between. This yields the analogous quantum result to Proposition \ref{prop:bound} which is used in the analysis of soundness of the two ZKPs presented in this paper. 

\section{A ZKP for the Subset Sum Problem}\label{sec:SSP}

\subsection{General Idea}

In this section, we'll introduce a two-prover ZKP of the Subset Sum Problem that relies on the homomorphic properties of the relativistic bit commitment. We'll start by stating the problem:

\textbf{Subset Sum Problem: } Given a set $\{s_1,...,s_n\}$ of positive integers, is there a subset that sums to $k$? A solution is a binary vector $v\in \mathbb{F}_2^n$ for which $\sum_{i=1}^n v_i s_i=k$. This version of the subset-sum problem is NP-complete \cite{10.5555/1051910}.

The intuition for why the following protocol works is quite simple. Imagine Alice is trying to convince Bob that the set $\{1,4,5,7,8\}$ has a subset that sums to $14$, without giving away the subset. She sets up two rows of $5$ upside-down cups on a table, and in the $i$th column, one cup has no marbles and the other cup has $s_i$ marbles. This is done randomly and independently so that Bob does not know which cup is empty in each column. At this stage, the upside-down cups may have an arrangement of marbles similar to Figure ~\ref{fig:marbleorig}.

\begin{figure}[ht!]
    \centering
    \includegraphics[scale=0.25]{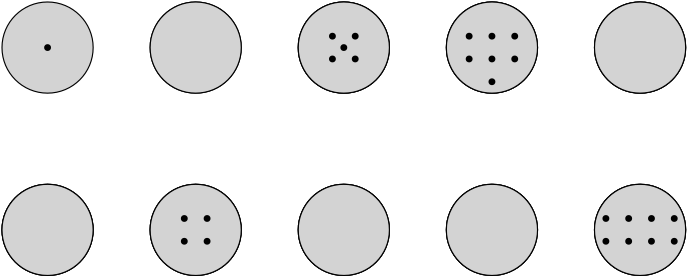}
    \caption{Original arrangement of marbles\label{fig:marbleorig}}
\end{figure}

Next, Bob can ask Alice either to lift each cup up one-by-one in order for him to check that Alice set it up properly, or he can ask her to slide one cup from each column to the edge of the table, then knock those cups into a bowl in order for him to count the $14$ marbles. This second option is represented in Figure ~\ref{fig:marble}. 

\begin{figure}[ht!]
    \centering
    \includegraphics[scale=0.2]{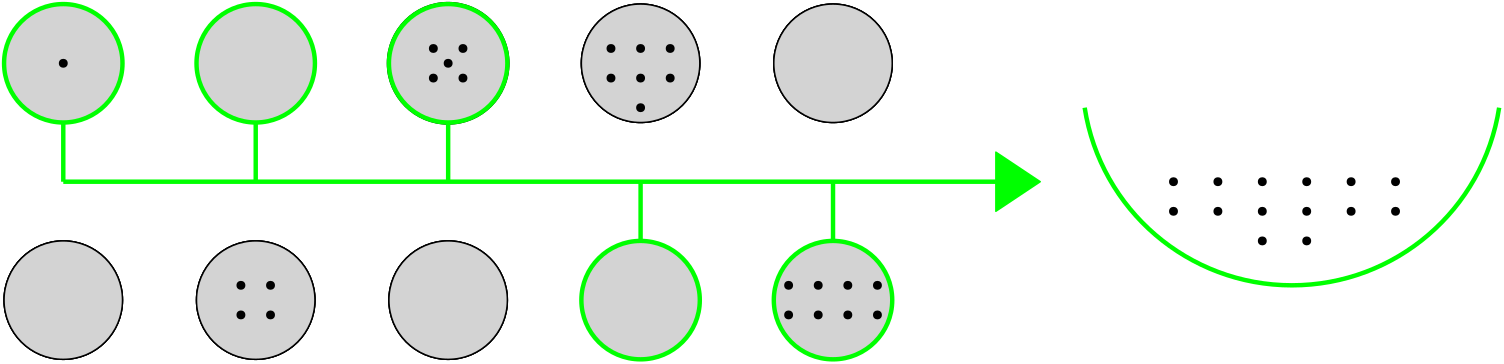}
    \caption{Action of Alice if Bob asks to see that the marbles sum to $k$\label{fig:marble}}
\end{figure}

Since the cups are knocked into the bowl quickly, he has no idea which marbles came from which cups, but he can count the number of marbles in total. By repeating this experiment multiple times with new random arrangements in each round, Bob will gain confidence that there is a subset that sums to $14$, but he will never get any clues to which subset it may be.

\subsection{Protocol}

We describe the notation in Table \ref{tab:SSPNot} that will appear in the protocol. The operations of the protocol will be over the field $\mathbb{F}_Q$ for some large prime $Q$ that exceeds the sum of all the elements in the set $s$. This way, the integers of the set $s$ and their sums can be interpreted as values of $\mathbb{F}_Q$. 

\begin{table}[h]
\centering
\begin{tabular}{||c  c  c||} 
 \hline
 \qquad Notation\qquad\qquad & \qquad\qquad Purpose \qquad\qquad & \qquad Data-Type \qquad \\ [0.5ex] 
 \hline\hline
 $c_0,c_1$ & keys for rows of cups & $\mathbb{F}^n_Q$\\
 \hline
 $w_0,w_1$ & encryption of rows of cups & $\mathbb{F}^n_Q$\\
 \hline
 $a$ & randomly chosen by V & $\mathbb{F}_Q$\\
 \hline
 $s$ &  input set & $\mathbb{F}_Q^n$\\
 \hline
 $k$ & subset sum target & $\mathbb{F}_Q$ \\
 \hline
 $v$ & solution to the problem & $\mathbb{F}^n_2$ \\
 \hline
 $x$ & indicates cups for solution & $\mathbb{F}^n_2$\\
 \hline
 $z$ & indicates empty cups & $\mathbb{F}^n_2$\\
 \hline
 $c'$ & sum of keys of cups of $x$ & $\mathbb{F}_Q$\\
 \hline
 $chall$ & randomly chosen by V & $\mathbb{F}_2$\\
 \hline
 $\cdot$ & scalar multiplication & operation\\
 \hline
 $*$ & entry-wise multiplication & operation\\
 \hline
\end{tabular}
\centering
\caption{Notation used in ZKP for Subset Sum}
\label{tab:SSPNot}
\end{table}

The operation $\cdot$ is scalar multiplication defined entry-wise: $a\cdot (x_1,...,x_n):= (ax_1,...,ax_n)$. We now introduce a ZKP between two provers and verifiers. Assuming the two provers know a witness $v$, they will be able to convince two verifiers that there exists a solution to the given Subset Sum Problem instance. Before engaging in a round of the protocol, P1 and P2 share random vectors $c_0,c_1\in\mathbb{F}_Q^n$ and $z\in\mathbb{F}_2^n$. The protocol is described in Table \ref{tab:SSPPro}. 

\begin{table}[h]
\centering
\begin{tabular}{||p{0.9\linewidth}||}  
 \hline
 Two-Prover, Two-verifier Subset Sum ZK protocol \\ [0.5ex] 
 \hline
    1. V1 sends P1 a random value $a\in\mathbb{F}_Q$.\\[1ex]
    2. P1 replies with $w_0=a\cdot (s*z)+c_0$ and $w_1=a\cdot (s* \overline{z})+c_1$.\\[1ex]
    3. V2 sends P2 $chall\in\{0,1\}$.\\[1ex]
    4. If $chall=0$, then P2 sends V2 $z,c_0,c_1$. If $chall=1$, then P2 will send to V2 the binary vector $x=v\oplus z$ and the value $c'=\sum\limits_{i=1}^n (c_{x_i})_i$.\\[1ex]
    5. After the round, if $chall=0$, then the verifiers confirm that $w_0=a\cdot (s* z)+c_0$ and $w_1=a\cdot (s* \overline{z})+c_1$. If instead $chall=1$, then the verifiers check that $\sum\limits_{i=1}^n (w_{x_i})_i=ak+c'$. If a check fails, they reject, otherwise they accept.\\[2ex] 
 \hline
\end{tabular}
\centering
\caption{ZKP for Subset Sum}
\label{tab:SSPPro}
\end{table}

\subsection{Proof of Security}

\begin{proposition}\label{prop:SSPComp}
    The ZK Subset Sum protocol has perfect completeness.
\end{proposition}

\begin{proof}
    Suppose the provers have a solution $v\in\mathbb{F}_2^n$ and random shared vectors $c_0,c_1\in\mathbb{F}_Q^n$ and $z\in\mathbb{F}_2^n$. Then it is clear that if steps (2) and (4) are followed properly by the provers, then the verification in (5) for $chall=0$ will pass. On the other hand, if $chall=1$, then the checks will still pass since we have
    \begin{align*}
        \sum\limits_{i=1}^n (w_{x_i})_i &=\sum\limits_{i=1}^n \overline{x_i}(w_0)_i+x_i (w_1)_i\\
        &=\sum\limits_{i=1}^n \overline{x_i}(a s_i z_i+(c_0)_i)+x_i (as_i \overline{z_i}+(c_1)_i)\\
        &= \sum\limits_{i=1}^n as_i (\overline{x_i}z_i + x_i \overline{z_i}) +(\overline{x_i}(c_0)_i +x_i (c_1)_i)\\
        &=\sum\limits_{i=1}^n as_i v_i+\sum\limits_{i=1}^n(\overline{x_i}(c_0)_i +x_i (c_1)_i)\\
        &=a\sum\limits_{i=1}^n s_i v_i+\sum\limits_{i=1}^n (c_{x_i})_i\\
        &=ak+c'.
    \end{align*}
    
    Therefore all the checks by the verifiers will pass, so the probability that the provers are accepted by the verifiers is 1.
    \qed    
\end{proof}

\begin{proposition}
    The Subset Sum ZKP is sound against malicious quantum provers with soundness exponentially close to $\frac{1}{2}$ in a single round.
\end{proposition}
The proof uses a technique similar to the proof of soundness in \cite{Chailloux_2017}. Here are the steps that will be taken in the formal proof.
\begin{itemize}
    \item Formalize the game that cheating provers play. This allows us to analyse $G_{coup}$ from Definition \ref{defin:gamecoup} where the second prover must answer both challenges at once. Note that $G_{coup}$ is not zero-knowledge.
    \item Since the Subset Sum Problem instance has no solution, combining the answers for both challenges must not yield a solution. Using this, we construct a strategy for P2 to guess $a$ using basic modular arithmetic, assuming P2 can successfully answer both challenges.
    \item By no-signalling, the probability of P2 successfully guessing $a$ is not more than $\frac{1}{Q}$. This yields an upper bound on the probability that P2 can answer both challenges, hence upper bounding $\omega^*(G_{coup})$. 
    \item We use Proposition~\ref{prop:bound1} from \cite{Chailloux_2017} to relate the winning probability of $G_{coup}$ with the winning probability of the ZKP for the Subset Sum Problem for cheating provers.\\
\end{itemize}

\begin{proof}
    We define the game $G^{SS}$ so that it satisfies Definition \ref{defin:game}.
    \begin{itemize}
        \item P1 receives value $a\in \mathbb{F}_Q$, and P2 receives $chall\in \{0,1\}$.
        \item P1 outputs values $w_0,w_1\in \mathbb{F}_Q^n$. If $chall=0$, then P2 outputs $z\in \mathbb{F}_2^n$ and $c_0,c_1\in \mathbb{F}_Q^n$. If $chall=1$, then P2 outputs a value $c'\in \mathbb{F}_Q$ and a vector $x\in \mathbb{F}_2^n$.
        \item If $chall=0$, then the two players win if $w_0=a\cdot (s * z)+c_0$ and $w_1=a\cdot (s* \overline{z})+c_1$. If $chall=1$, then players win if $\sum\limits_{i=1}^n (w_{x_i})_i=ak+c'$. 
    \end{itemize}
    
    Recall Definition \ref{defin:proj}. The game $G^{SS}$ is $2^n$-projective since after $z$ or $x$ is chosen by P2, then the winning values of $c_1,c_2$ or $c'$ are fixed.  This can be seen by rearranging the equations in the last bullet-point. In order to upper bound $\omega^*(G^{SS})$, we consider the game $G^{SS}_{coup}$, given in Definition \ref{defin:gamecoup}. Our goal is to show that if cheating provers can win the game $G^{SS}_{coup}$, then P2 has a strategy to perfectly guess $a$, which should only happen with negligible probability. Fix an input/output pair $(a,(w_0,w_1))$ for P1, and consider the outputs for P2 for both inputs. For $chall=0$, we have
    \begin{align*}
        w_0=a\cdot s*z+c_0\text{  and  }w_1=a\cdot s*\overline{z}+c_1\\
        \implies w_b=a(\overline{b}\cdot s*z+b\cdot s*\overline{z})+c_b.
    \end{align*}
    
    For $chall=1$, we have 
    \begin{align*}
        \sum\limits_{i=1}^n (w_{x_i})_i=ak+c'.
    \end{align*} 
    
    Substituting the left hand side, then rearranging, we get
    \begin{align*}
        ak+c'=\sum\limits_{i=1}^n (w_{x_i})_i=\sum\limits_{i=1}^n a\cdot (\overline{x_i}s_iz_i+x_i s_i\overline{z_i})+(c_{x_i})_i \\
        \implies a\left(\sum\limits_{i=1}^n (\overline{x_i}s_iz_i+x_i s_i\overline{z_i})-k\right) =c'-\sum\limits_{i=1}^n (c_{x_i})_i.\\
    \end{align*}
    
    CLAIM: $\sum\limits_{i=1}^n (\overline{x_i}s_iz_i+x_i s_i\overline{z_i})\neq k$.\\
    
    Assume by contradiction that we have equality. Then we will construct a solution to the Subset Sum Problem, however since we are proving soundness in the case of a dishonest prover, we have implicitly assumed this to not be possible.
    The candidate solution is $v':=x\oplus z$. Indeed, we have 
    \begin{align*}
        \sum_{i=1}^n v'_i s_i= \sum_{i=1}^n (x\oplus z)_i s_i= \sum_{i=1}^n (\overline{x_i} z_i+x_i\overline{z_i} )s_i=\sum_{i=1}^n \overline{x_i} z_i s_i + x_i\overline{z_i}s_i=k.
    \end{align*}
    
    Therefore we have constructed a solution $v'$, hence $k$ cannot be equal to the sum.\\
    
    By the result of the claim, we can divide and obtain
    \begin{align*}
        \implies a &=\left(c'-\sum\limits_{i=1}^n (c_{x_i})_i\right)\left(\sum\limits_{i=1}^n (\overline{x_i}s_iz_i+x_i s_i\overline{z_i})-k\right)^{-1}.
    \end{align*}
    
    However, since the value on the right hand side is completely determined by the output of P2 and the value on the lefthand side is chosen uniformly at random for P1, this can happen with probability at most $\frac{1}{Q}$ by no-signalling. This means that $\omega^*(G^{SS}_{coup})\leq \frac{1}{Q}$. Now, we can apply Proposition \ref{prop:bound1} from \cite{Chailloux_2017} to obtain 
    \begin{align*}
        \omega^*(G^{SS})\leq \frac{1}{2}+\left(\frac{64\cdot 2^n}{Q}\right)^{1/3}.
    \end{align*}
    
    If we take $Q\geq 64\cdot 2^{n+3K}$, then the protocol has soundness $\frac{1}{2}+2^{-K}$. Lastly, we verify that there is only a polynomial amount of communication needed for the protocol. In (1), V1 sends $\log(Q)$ bits, then in (2), P1 sends $2n\log(Q)$ bits, then in (3), only 1 bit is sent, and finally in (4), $n+2n\log(Q)$ bits are sent if $chall=0$, and if $chall=1$, then $n+\log(Q)$ bits are sent. Therefore each step requires only a polynomial number of bits in $n$ and $K$ since $\log(Q)\in O(n+K)$.
    \qed
\end{proof}

\begin{proposition}
    The Subset Sum ZKP has perfect zero knowledge against malicious quantum verifiers.
\end{proposition}

\begin{proof}
    We'll show that this protocol is zero knowledge in the model of 2-provers and a single quantum verifier. We will give the verifier the freedom to query the provers in any order, without needing to respect relativistic timing restraints. Proving the ZK property in this model will be even stronger since the model gives as much power to the malicious verifier as possible. We model a cheating verifier $V^*$ as two families of circuits $(V_1^*,V_2^*)$, where $V_i^*$ takes as input a sub-view and outputs the message to prover $i$ for an instance of size $n$. Since the verifier can ask the questions in either order, we have two cases:
    \begin{itemize}
        \item Case 1: $V_2^*$ depends on the interaction with P1.
        \item Case 2: $V_1^*$ depends on the interaction with P2.
    \end{itemize}
    
    Since the cases are treated very similarly, we will just present the proof of Case 1. In other words, we'll consider a verifier that queries P1 and waits for a response before querying P2. Our first step is to describe the view when $V^*$ interacts with two honest provers. Then, we will define a simulator that can create the same view using only query access to $V^*$ despite not having any access to provers. Since these two views will be the same, then the protocol will be zero-knowledge.
    
    The view will consist of classical registers $Q_1$ and $Q_2$ that will hold the questions that will be asked to P1 and P2. Also, the responses of the two provers will be stored in classical registers $R_1$ and $R_2$. In addition, the verifier will hold a private quantum register $V$. In addition, we will adopt the notation $D(\psi):=|\psi\rangle\langle\psi|$ for quantum states to avoid the need to write $\psi$ twice.\\
    
    \textbf{Case 1 with Honest Prover}
     
    We assume that the operation of $V_2^*$ will depend on the interaction with P1.
    At the beginning of the protocol, the verifier's view is an auxiliary state $\sigma_0:=\rho_V$. 
    
    Next, after the verifier's first message, the view is
    \begin{align*}
        \sigma_1:=V_1^*(\rho_V)=\sum_{a\in \mathbb{F}_Q} p_{a}D(a)_{Q_1}\otimes \rho(a)_V.
    \end{align*}

    Here, $p_a$ is the probability of $a$ being the query for P1, and $\rho(a)_V$ is the verifier's private quantum state after sending $a$ to P1. Following the response of P1, the verifier's view is

    \begin{align*}
        \sigma_2:=\frac{1}{Q^{2n}}\frac{1}{2^n}\sum_{c_0,c_1\in \mathbb{F}_Q^n}\sum_{z\in \mathbb{F}_2^n}\sum_{a\in \mathbb{F}_Q} p_{a} D(Y(z,c_0,c_1))_{R_1}\otimes      D( a)_{Q_1}\otimes \rho(a)_V.
    \end{align*}
    where $Y(z,c_0,c_1):=(a\cdot s* z+c_0,a\cdot s* \overline{z}+c_1)$. Next, the verifier sends the challenge, which can be influenced by all that has happened up to this point. In other words, the verifier applies the circuit $V_2^*$ to the view $\sigma_2$. The view becomes 
    \begin{align*}
        \sigma_3:=\frac{1}{Q^{2n}}\frac{1}{2^n} &\sum_{c_0,c_1\in \mathbb{F}_Q^n}\sum_{z\in \mathbb{F}_2^n}\sum_{a\in \mathbb{F}_Q} \sum_{chall\in\{0,1\}} p_{a,chall}  D(Y(z,c_0,c_1))_{R_1}\\
        &\otimes D(chall)_{Q_2} \otimes D(a)_{Q_1}\otimes \rho(a,chall,Y(z,c_0,c_1))_V.
    \end{align*}
    
    After the final message, we have two cases. If the challenge is 0, then P2 sends $z,c_0,c_1$ in the second response register. On the other hand, if the challenge is 1, then that register should instead contain a value $c'$ and a binary vector $x$. In particular, an honest prover will choose $x=v\oplus z$ and $c'=\sum\limits_{i=1}^n (c_{x_i})_i$, yielding
    \begin{align*}
        \sigma_4:=\frac{1}{Q^{2n}}\frac{1}{2^n} &\sum_{c_0,c_1\in \mathbb{F}_Q^n}\sum_{z\in \mathbb{F}_2^n}\sum_{a\in \mathbb{F}_Q} D(Y(z,c_0,c_1))_{R_1} \otimes D(a)_{Q_1}\otimes\\
        &\Big(p_{a,0}D(0)_{Q_2} \otimes D(z,c_0,c_1)_{R_2} \otimes \rho(a,0,Y(z,c_0,c_1))_V\\
        &+p_{a,1}D(1)_{Q_2} \otimes D(v\oplus z,c')_{R_2} \otimes \rho(a,1,Y(z,c_0,c_1))_V\Big).
    \end{align*}
    
    This is the final view for an honest prover. Now, before moving onto the simulator, we will rewrite this final state so that it will resemble the simulator's final state later on. The first step is to let $(w_0,w_1):=Y(z,c_0,c_1)$ and sum over $w_0,w_1$ instead of $c_0,c_1$. This means that $c_0=w_0-a\cdot s*z$ and $c_1=w_1-a\cdot s*\overline{z}$. So far, this gives us
    \begin{align*}
        \sigma_4=\frac{1}{Q^{2n}} &\frac{1}{2^n} \sum_{w_0,w_1\in \mathbb{F}_Q^n}\sum_{z\in \mathbb{F}_2^n}\sum_{a\in \mathbb{F}_Q} D(w_0,w_1)_{R_1} \otimes D(a)_{Q_1}\otimes\\
        &\Big(p_{a,0}D(0)_{Q_2} \otimes D(z,w_0-a\cdot s*z,w_1-a\cdot s*\overline{z})_{R_2} \otimes \rho(a,0,w_0,w_1)_V\\
        &+p_{a,1}D(1)_{Q_2} \otimes D(v\oplus z,c')_{R_2} \otimes \rho(a,1,w_0,w_1)_V\Big).\\
    \end{align*}
    
    The next step is to move the sum over $z$ past the terms on which it does not act. We obtain
    \begin{align*}
        \sigma_4=\frac{1}{Q^{2n}} & \sum_{w_0,w_1\in \mathbb{F}_Q^n}\sum_{a\in \mathbb{F}_Q} D(w_0,w_1)_{R_1} \otimes D(a)_{Q_1}\otimes \Big(p_{a,0}|0\rangle\langle 0|_{Q_2} \otimes\\
        &\frac{1}{2^n} \sum_{z\in \mathbb{F}_2^n} D(z,w_0-a\cdot s*z,w_1-a\cdot s*\overline{z})_{R_2} \otimes \rho(a,0,w_0,w_1)_V\\
        &+p_{a,1}D(1)_{Q_2} \otimes\frac{1}{2^n} \sum_{z\in \mathbb{F}_2^n} D(v\oplus z,c')_{R_2} \otimes \rho(a,1,w_0,w_1)_V\Big).\\
    \end{align*}
    
    The final step is to rename $x:=v\oplus z$ in the third line and sum over $x$ instead. We can also use the relation from Proposition ~\ref{prop:SSPComp} to obtain $c'=\sum\limits_{i=1}^n (w_{x_i})_i-k a$. Putting this together, we get 
    \begin{align*}
        \sigma_4=\frac{1}{Q^{2n}} & \sum_{w_0,w_1\in \mathbb{F}_Q^n}\sum_{a\in \mathbb{F}_Q} D(w_0,w_1)_{R_1} \otimes D(a)_{Q_1}\otimes \Big(p_{a,0}D(0)_{Q_2} \otimes\\
        &\frac{1}{2^n} \sum_{z\in \mathbb{F}_2^n} D(z,w_0-a\cdot s*z,w_1-a\cdot s*\overline{z})_{R_2} \otimes \rho(a,0,w_0,w_1)_V\\
        &+p_{a,1}D(1)_{Q_2} \otimes\frac{1}{2^n} \sum_{x\in \mathbb{F}_2^n} D(x,\sum\limits_{i=1}^n (w_{x_i})_i-k a)_{R_2} \otimes \rho(a,1,w_0,w_1)_V\Big).\\
    \end{align*}
    
    \textbf{Case 1 with Simulator}
    
    Now, we describe how to simulate the views of the verifier without the help of any provers. We will denote the $i$th simulated view as $\sigma'_i$. Since the simulator has access to $V_1^*$ and $\rho_V$, then $\sigma_0$ and $\sigma_1$ are straightforward to simulate. Intuitively, no effort is required at this stage because the prover has not acted yet. For $\sigma_2$, the response from P1 is two uniformly random vectors $w_0,w_1$ since $c_0$ and $c_1$ act as one-time pads. Then
    \begin{align*}
        \sigma_2'=\frac{1}{Q^{2n}}\sum_{w_0,w_1\in \mathbb{F}_Q^n}\sum_{a\in \mathbb{F}_Q} p_{a}  D(w_0,w_1)_{R_1}\otimes      D(a)_{Q_1}\otimes \rho(a)_V.
    \end{align*}
    
    This can be created from $\sigma_1'$ by tensoring with the maximally mixed state in register $R_1$. Next, the simulator applies $V_2^*$ to $\sigma_2'$ to get $\sigma_3'$:
    \begin{align*}
        \sigma_3'=\frac{1}{Q^{2n}} &\sum_{w_0,w_1\in \mathbb{F}_Q^n}\sum_{a\in \mathbb{F}_Q} \sum_{chall\in\{0,1\}} p_{a,chall}  D(w_0,w_1)_{R_1}\\
        &\otimes D(chall)_{Q_2} \otimes D(a)_{Q_1}\otimes \rho(a,chall,w_0,w_1)_V.
    \end{align*}
    
    The final step is to simulate $\sigma_4$. Again, this depends on the challenge bit $chall$. 
    \begin{itemize}
        \item For $chall=0$ in register $Q_2$, the simulator chooses a random $z\in \mathbb{F}_2^n$ and then puts $D(z,w_0-a\cdot s * z,w_1-a\cdot s * \bar{z})$ in register $R_2$.
        \item For $chall=1$ in register $Q_2$, the simulator chooses a random $x\in \mathbb{F}_2^n$ and then places $D(x, \sum\limits_{i=1}^n (w_{x_i})_i-k a)$ in register $R_2$.
    \end{itemize}
    
    This gives a final view of 
    \begin{align*}
        \sigma_4'&= \frac{1}{Q^{2n}} \sum_{w_0,w_1\in \mathbb{F}_Q^n}\sum_{a\in \mathbb{F}_Q^n}   D(w_0,w_1)_{R_1}\otimes D(a)_{Q_1}\otimes\\
        &\Big(p_{a,0}D(0)_{Q_2}\otimes \frac{1}{2^n}\sum_{z\in \mathbb{F}_2^n}D(z,w_0-a\cdot s * z,w_1-a\cdot s * \bar{z})_{R_2}\otimes\rho(a,0,w_0,w_1)_V\\
        &+p_{a,1}D(1)_{Q_2}\otimes \frac{1}{2^n}\sum_{x\in \mathbb{F}_2^n}D(x, \sum\limits_{i=1}^n (w_{x_i})_i-k a)_{R_2}\otimes\rho(a,1,w_0,w_1)_V\Big).\\
    \end{align*}

    Now, it is clear that the two views $\sigma_4$ and $\sigma_4'$ are equal, hence our simulator succeeded.
    
    \textbf{Case 2}
    The proof of Case 2 is very similar and is left out. 
    \qed
\end{proof}

\subsection{Analysis of Efficiency}

Now we analyze the efficiency of the Subset Sum ZKP. The best known quantum algorithm for the Subset Sum Problem has time complexity $O(2^{n/3})$ \cite{https://doi.org/10.4230/lipics.esa.2022.6}. In order to take roughly time $2^{100}$ to solve, the instance must therefore be of size $n\geq 300$. The prime $Q$ is chosen such that $Q>64\cdot 2^{n+3K}$ with $K=5$ so that soundness is $\frac{1}{2}+2^{-K}\approx 0.53$. Then the expected number of bits sent in each round is 
\begin{align*}
    \log_2(Q)+1+2n\log_2(Q)+(n+\frac{2n\log_2(Q)+\log_2(Q)}{2})\approx 290,000   
\end{align*} 
where the terms correspond to the number of bits sent by V1, V2, P1, P2, in that order. This means that each round requires roughly 36KB of communication. This is far less than the 1.89MB required by \cite{Chailloux_2017}, and comparable to the 17KB required by the protocol of \cite{https://doi.org/10.48550/arxiv.2112.01386}. Note however that the Subset Sum ZKP has a lower soundness error than \cite{https://doi.org/10.48550/arxiv.2112.01386} (roughly $\frac{1}{2}$ compared to $\frac{2}{3}$), meaning that less rounds are required. To reduce the total soundness error to $2^{-100}$, our Subset Sum ZKP would require only 110 rounds $(0.53^{110}\approx 2^{-100})$, whereas 
\cite{https://doi.org/10.48550/arxiv.2112.01386} would require 170 rounds $(0.67^{170}\approx 2^{-100}$). Putting this together, the total information sent in our protocol is roughly 3.96MB (36KB per round with 110 rounds), compared to 2.89MB of total communication for \cite{https://doi.org/10.48550/arxiv.2112.01386} (17KB per round with 170 rounds). This means that our ZKP is very comparable to \cite{https://doi.org/10.48550/arxiv.2112.01386} in terms of total information sent. 

One important consideration is the complexity of computation for the provers. We will first motivate this by explaining the experimental setup, which is depicted in Figure \ref{fig:distance}. 

\begin{figure}[ht!]
    \centering
    \includegraphics[scale=0.22]{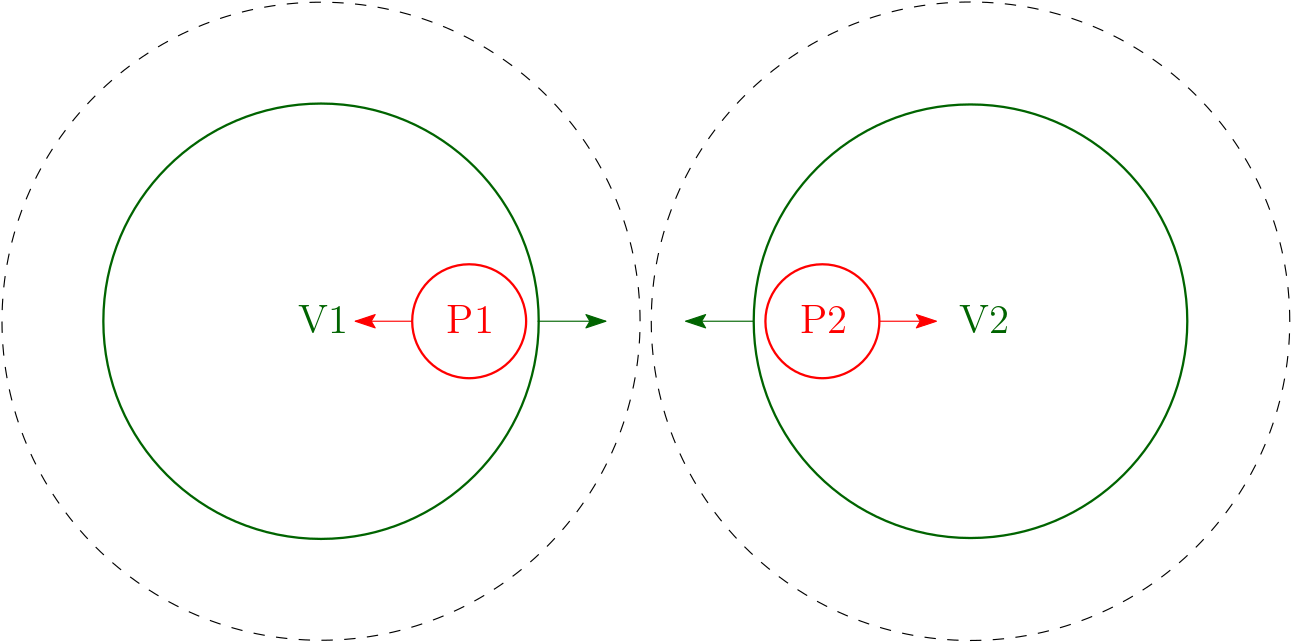}
    \caption{Message propagation in experimental setup \label{fig:distance}}
\end{figure}

Before the rounds, the verifiers position themselves a known distance apart. Then in every round, they each send their message to their respective prover. The provers each receive a message, perform some computation, then reply back to their respective verifier. In order to ensure that the provers are separated, the verifiers time the responses to ensure that the provers reply within fixed time limits. These time limits depend on the distance between the verifiers, and ensure that the provers are inside the dotted circles, hence separated. Importantly, the response of P2 is not influenced by the question of V1 and similarly, the response of P1 is not influenced by the question of V2. For example, after V1 sends the message $a$ intended for P1, then if P2 receives $a$ as well, then P2 will still not have enough time to reply to V2 with a message influenced by $a$. This is crucial to fulfill the assumption that each prover only has access to the question intended for them. For more details on the experimental setup, see \cite{Alikhani_2021} or \cite{https://doi.org/10.48550/arxiv.2112.01386}.

In theory, if all the messages travel towards the other parties at the speed of light, and each party can instantly reply when they receive a message, then as long as the provers stay within the dotted circles, then the protocol should run smoothly. However in practice, the messages may not follow such a direct path to their intended recipients, and the provers must take some time to perform internal computations once they have received their message before they can reply. Effectively, this means that the provers cannot be close to the edges of the dotted circles since it will not give them enough time to respond, and slower protocols require the verifiers to further increase their separation. Experimentally, large distances have been required by the verifiers in order to give the provers enough time to respond. For example, the small-scale setup in \cite{https://doi.org/10.48550/arxiv.2112.01386} required the verifiers to position themselves 400km apart, whereas in \cite{Alikhani_2021}, the verifiers were only 60m apart. These differences are a result of the amount of communication required in each round and the time complexity of the operations required by the provers. For this reason, we also seek to speed up the computations required by the honest provers in order to allow for more practical, small-scale uses of ZKPs. 

With the experimental setup explained, we can now compare multiplication complexities. While only three multiplications must be performed by the prover in \cite{https://doi.org/10.48550/arxiv.2112.01386}, their field $\mathbb{F}_Q$ is far larger than ours, chosen to have $2^{23209}-1$ elements. To their benefit, they only need to perform three multiplications per round, whereas in our protocol, $2n$ multiplications must be performed, albeit in a much smaller field of size roughly $2^{321}$. As an added bonus, our multiplications are all by the same constant $a$, hence parallelism is very straightforward. 

Multiplication can have varying time complexities depending on the implementation, and some algorithms contain large hidden constants. To simplify the comparison, tests were run using Python with and without parallelism (array multiplication in NumPy). In our brief experiments, our multiplications were roughly 3 times faster without parallelism and roughly 7 times faster with parallelism when compared to \cite{https://doi.org/10.48550/arxiv.2112.01386}. We leave a more thorough analysis for future work. 

In short, faster multiplication can allow the verifiers in the protocol to be at closer distances, making it more practical, and allowing the protocol to run faster. In Table \ref{tab:efficiency} we present a summary of the results by comparing our Subset Sum ZKP with other known quantum-secure ZKPs in the literature.

\begin{table}[h]
\centering
\begin{tabular}{||c  c  c  c||} 
 \hline
 \qquad Protocol \qquad & \qquad $\#$Bytes$/$Round \qquad & \qquad $\#$Rounds \qquad & \qquad Multiplication complexity \\ [0.5ex] 
 \hline\hline
 \cite{Alikhani_2021} & 2B & $10^{19}$ & negligible\\
 \hline
 \cite{Chailloux_2017} & 1.89MB & 177 & $40$ to $300\times$ slower\\
 \hline
 \cite{https://doi.org/10.48550/arxiv.2112.01386} & 17.03KB & 170 & $3$ to $7\times$ slower\\
 \hline
 Subset Sum ZKP & 36KB & 110 & Baseline\\
 \hline
\end{tabular}
\centering
\caption{Comparison of efficiency of known ZKPs}
\label{tab:efficiency}
\end{table}

Note that although the protocol of \cite{Alikhani_2021} has very efficient rounds, three provers are required to achieve quantum security, and the number of rounds is impractical. Note that the numbers presented here differ slightly from in \cite{Chailloux_2017}\footnote{ \cite{Chailloux_2017} requires 177 rounds because in \cite{https://doi.org/10.48550/arxiv.2112.01386}, $Q$ is chosen to be $10000n!=64n!2^{3k}$ so $k\approx 2.4$, hence soundness is $\frac{1}{2}+2^{-2.4}\approx 0.69$}  and in \cite{https://doi.org/10.48550/arxiv.2112.01386}\footnote{In \cite{https://doi.org/10.48550/arxiv.2112.01386}, 340 rounds were chosen for loss tolerance}.

One last point is that generating positive instances of the Subset Sum Problem is straightforward by sampling $n$ random numbers between 1 and $Q/n$, then determining the target $k$ by randomly choosing a subset of the $n$ numbers.

\section{A ZKP for 3-SAT}\label{sec:3SAT}

\subsection{General Idea}

3-SAT problem: Given a Boolean expression $\phi$ of the variables $x_1,...,x_n$ in conjunctive normal form with $m$ clauses of size $3$ (3-CNF), determine whether there is an assignment for the variables that satisfies $\phi$. For example, the variables may be $x_1,...,x_5$, and $\phi$ may be 
\begin{align*}
    \phi'\equiv (x_3\lor \neg x_2\lor x_5)\land (\neg x_1\lor \neg x_4\lor \neg x_5)\land (x_1\lor \neg x_2\lor x_5)\land (x_1\lor x_4\lor x_2).
\end{align*}

Then $\phi'$ has a satisfying assignment of $x=(1,0,1,0,0)$. This problem is known to be NP-complete \cite{AhoAlfredV1974Tdaa}.

Note that a witness to the a 3SAT instance is typically viewed as an assignment $x$ of the variables, but instead for this section, we will construct a witness differently. An alternate way to present a solution to the 3SAT problem is to provide a vector $e\in \{1,2,3\}^m$ that indicates the position of a 1 in each clause of $\phi$. This way, one could perform a linear scan through $\phi$ and determine the variable assignments. The only required check would be that the same variable does not obtain two different assignments based on two different clauses. For example, an $e\in \{1,2,3\}^4$ for $\phi'$ could be $e':=(1,2,1,1)$, though it is not unique. 

Many ZKPs require a randomized step to ensure that the verifier does not learn anything about the witness when the prover unveils an answer. Here, the randomization will be independent cyclic permutations of the variables in each clause.

\begin{definition}
    Let $\phi$ be a CNF with $m$ clauses, each of size 3. We define a cyclic CNF permutation $\Pi$ to be a collection of $m$ cyclic permutations on $\{1,2,3\}$. By $\Pi(\phi)$, we denote $\phi$ after each clause has been permuted by the corresponding permutation of $\Pi$. We define $CP$ to be the set of all cyclic CNF permutations on $\phi$.
\end{definition}

For example by choosing a random $\Pi\in CP$, with $\phi'$ as above, then $\Pi(\phi')$ could be the following:

\begin{align*}
    \Pi(\phi')\equiv (x_5\lor x_3\lor \neg x_2)\land (\neg x_4\lor \neg x_5\lor \neg x_1)\land (x_5\lor x_1\lor \neg x_2)\land (x_1\lor x_4\lor x_2).
\end{align*}

Note that simply re-arranging the contents of the CNF does not change the satisfiability of the formula since $\lor$ is commutative. In fact, the values of $e\in\{1,2,3\}^m$ must get adjusted according to $\Pi$. We write $\Pi(e)$ to denote the witness in $\{1,2,3\}^m$ for $\Pi(\phi)$ where each coordinate of $e$ has undergone the corresponding permutation in $\Pi$. For example, $\Pi(e')=(2,1,2,1)$. The idea with the protocol below is that the provers will commit to the satisfying assignment, as well as the formula $\Pi(\phi)$ with the satisfying assignment substituted in for the variables. For example, with the running example of $\phi'$, the provers would make the following commitment:

\begin{figure}[ht!]
    \centering
    \includegraphics[scale=0.22]{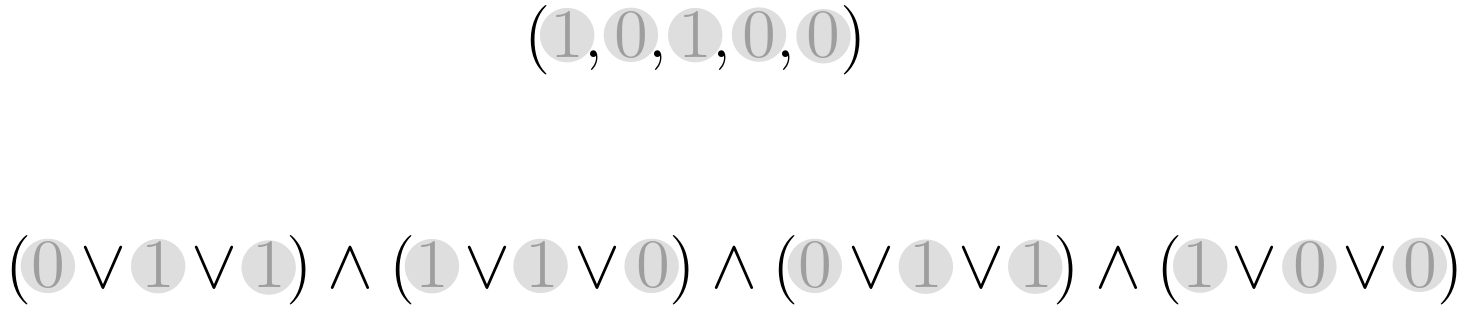}
    \caption{Commitment by provers\label{fig:3sat1}}
\end{figure}

Next, the verifiers can either challenge the provers to unveil a 1 in each clause, or to prove that the variable assignments were consistent. If the challenge was the former, then the verifiers would see Figure \ref{fig:3sat2}

\begin{figure}[ht!]
    \centering
    \includegraphics[scale=0.22]{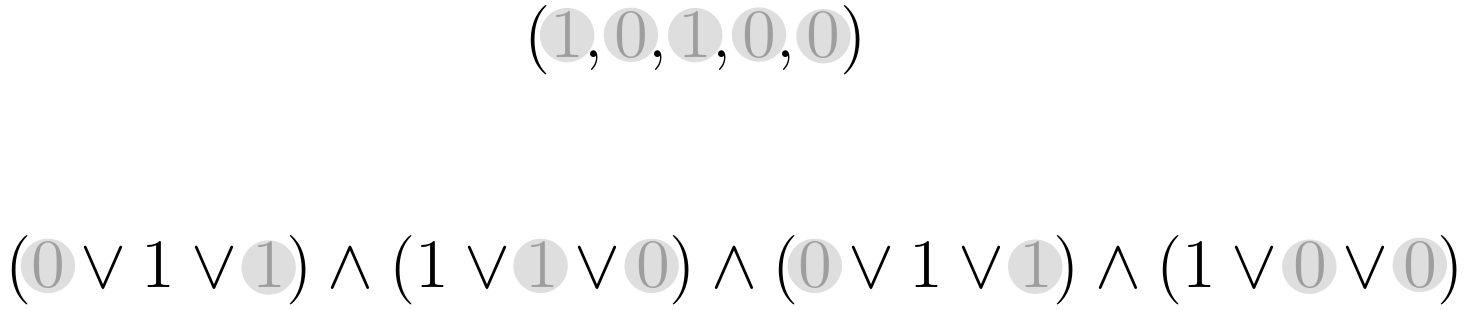}
    \caption{Provers unveil a 1 in each clause\label{fig:3sat2}}
\end{figure}

Note that the unveiled 1's are in the positions specified by $\Pi(e')$. On the other hand, if the verifiers asked to see that the variable assignments were consistent, then the provers would go through each variable of $\Pi(\phi)$ and prove that it's value is consistent with the satisfying assignment. Figure \ref{fig:3sat3} depicts this step for only the third clause to avoid clutter.

\begin{figure}[ht!]
    \centering
    \includegraphics[scale=0.13]{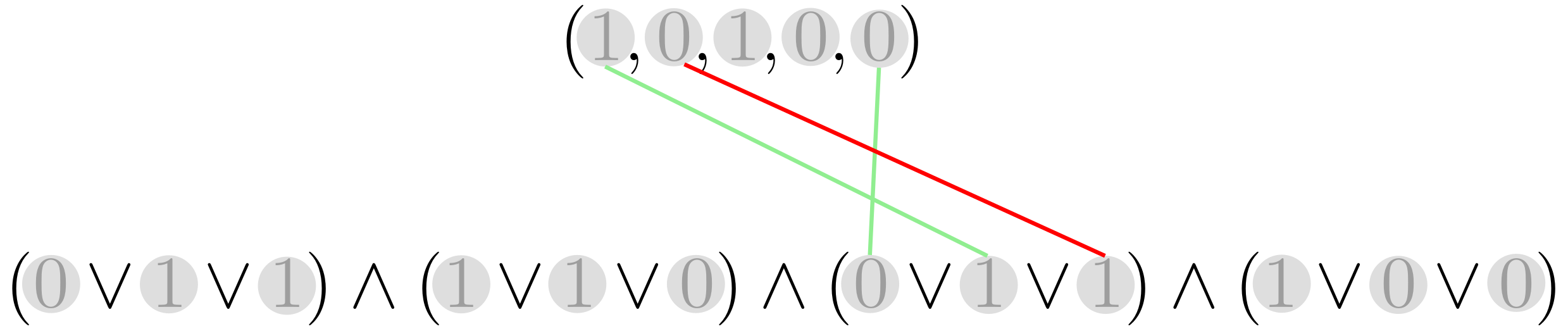}
    \caption{Provers show consistent variable assignments\label{fig:3sat3}}
\end{figure}

Note that the provers will use their ability to unveil the difference of commitments for the green lines, and use their ability to unveil the sum of commitments for the red lines. In particular, the difference of two equal bits is 0, and the sum of two different bits is 1. 

\subsection{Protocol}

First, we introduce notation that will be used throughout the protocol in Table \ref{tab:3SATNot}. The operations will be mainly addition and multiplication over the field $\mathbb{F}_Q$, however, we will also multiply elements of $\mathbb{F}_Q$ by bits, hence $Q$ may be any large prime power since every field has elements 0,1.

\begin{table}[h]
\centering
\begin{tabular}{||c  c  c||} 
\hline
\qquad Notation\qquad\qquad & \qquad\qquad Purpose \qquad\qquad & \qquad Data-Type \qquad \\ [0.5ex] 
 \hline\hline
 $\Pi$ & cyclic 3-CNF permutation & CP\\
 \hline
 $\phi$ & 3-CNF formula & CP\\
 \hline
 $s'$ & satisfying assignment of $\phi$ & $\mathbb{F}_2^n$\\
 \hline
 $p$ & bits of permuted formula $\Pi(\phi)$ & $\mathbb{F}_2^{3m}$\\
 \hline
 $c$ & key for $p$ & $\mathbb{F}_Q^{3m}$\\
 \hline
 $c'$ & key for $s'$ & $\mathbb{F}_Q^n$\\
 \hline
 $w$ & encryption of $p$ & $\mathbb{F}_Q^{3m}$\\
 \hline
 $w'$ & encryption of $s'$ & $\mathbb{F}_Q^{n}$\\
 \hline
 $e$ & indicates a 1 in each clause of $\phi$ & $\{1,2,3\}^m$\\
 \hline
 $chall$ & random challenge bit chosen by verifier & $\{0,1\}$\\
 \hline
\end{tabular}
\centering
\caption{Notation used in ZKP for 3SAT}
\label{tab:3SATNot}
\end{table}

Next, we introduce a ZKP between two provers and two verifiers. Assuming the two provers know a witness $s'\in \{0,1\}^n$, they will be able to convince two verifiers that there exists a solution to the given 3SAT problem instance. First, they compute an $e\in \{1,2,3\}^m$ from $s'$ that indicates the position of a 1 in each clause. Then, before each round, the provers share random $\Pi\in CP$, $c\in \mathbb{F}_Q^{3m}$, $c'\in \mathbb{F}_Q^{n}$. The protocol is given in Table \ref{tab:3SATPro}.

\begin{table}[h]
\centering
\begin{tabular}{||p{0.9\linewidth}||}  
 \hline
 Two-Prover, Two-verifier 3-SAT ZKP \\ [0.5ex] 
 \hline
    1. V1 sends P1 a random value $a\in\mathbb{F}_Q$.\\[1ex]
    2. P1 replies with $w'=a\cdot s' +c'$ and  $w=a\cdot p +c$.\\[1ex]
    3. V2 sends P2 $chall\in\{0,1\}$.\\[1ex]
    4. If $chall=0$, P2 first sends $\Pi$ to V2. Then for each $i=1,...,3m$: let $x_j$ be at position $i$ of $\Pi(\phi)$. P2 sends to V2:
    \begin{itemize}
        \item $c_{i}+c'_j$ if variable $x_j$ is negated at position $i$.
        \item $c_{i}-c'_j$ if variable $x_j$ is not negated at position $i$.
    \end{itemize}
    Instead if $chall=1$, then for each $i=1,...,m$, P2 sends $(\Pi(e)_i,c_{3(i-1)+\Pi(e)_i})$ to V2, unveiling a 1 in each clause. \\[1ex]
    5. After the round, if $chall=0$, then the verifiers confirm that for each $i=1,...,3m$, if variable $x_j$ is at position $i$ of $\Pi(\phi)$, then
    \begin{itemize}
        \item $w_i+w'_j=a+(c_{i}+c'_j)$ if variable $x_j$ is negated at position $i$.
        \item $w_i-w'_j=c_{i}-c'_j$ if variable $x_j$ is not negated at position $i$.
    \end{itemize}
    If instead $chall=1$, then the verifiers check that for each $i=1,...,m$, it is the case that $w_{3(i-1)+\Pi(e)_i}= a+c_{3(i-1)+\Pi(e)_i}$.\\[2ex] 
 \hline
\end{tabular}
\centering
\caption{ZKP for 3SAT}
\label{tab:3SATPro}
\end{table}
~\\
~\\

\subsection{Proof of Security}

\begin{proposition}
    The 3-SAT ZKP has perfect completeness.
\end{proposition}

\begin{proof}
    Assume that the two honest provers have a witness $s'$. Then if they follow the steps in the protocol with two honest verifiers, then when $chall=0$, the checks that the verifiers make will pass. Indeed, if $i\in\{1,...,3m\}$ and variable $x_j$ is negated at position $i$ in $\Pi(\phi)$, then $s'_j+p_i=1$, yielding
    
    \begin{align*}
        w_i+w'_j=(ap_i+c_i)+(as'_j+c'_j)= a(p_i+s'_j)+(c_i+c'_j)=a+(c_i+c'_j).
    \end{align*}
    
    If instead variable $x_j$ is not negated at position $i$ in $\Pi(\phi)$, then $s'_j=p_i$, meaning that
    
    \begin{align*}
        w_i-w'_j=(ap_i+c_i)-(as'_j+c'_j)= a(p_i-s'_j)+(c_i-c'_j)=c_i-c'_j.
    \end{align*}
    
    Next we confirm that the checks will pass when $chall=1$. The provers are honest, meaning each clause contains a 1, therefore $p_{3(i-1)+\Pi(e)_i}=1\; \forall i\in\{1,...,m\}$, hence
    \begin{align*}
        w_{3(i-1)+\Pi(e)_i}= ap_{3(i-1)+\Pi(e)_i}+c_{3(i-1)+\Pi(e)_i}=a+c_{3(i-1)+\Pi(e)_i}.
    \end{align*}
    
    Therefore all the checks by the verifiers will pass for both values of $chall$.
    \qed
\end{proof}

\begin{proposition}
    The 3-SAT ZKP is sound against malicious quantum provers with soundness exponentially close to $\frac{1}{2}$ in a single round.
\end{proposition}

\begin{proof}
    We define the game $G^{3SAT}$. 
    \begin{itemize}
        \item P1 receives value $a\in \mathbb{F}_Q$, and P2 receives $chall\in \{0,1\}$.
        \item P1 outputs values $w\in \mathbb{F}_Q^{3m}, w'\in \mathbb{F}_Q^{n}$. If $chall=0$, then P2 outputs a $\Pi\in CP$ and a value $\delta \in \mathbb{F}^{3m}_Q$.
        If $chall=1$, then P2 outputs $f\in\{1,2,3\}^m$, $\gamma \in \mathbb{F}_Q^m$.
        \item After the round, if $chall=0$, then the provers win if for each $i=1,...,3m$, it holds that
        \begin{itemize}
            \item $w_i+w'_j=a+\delta_i$ where variable $x_j$ appears negated at position $i$ of $\Pi(\phi)$,
            \item $w_i-w'_j=\delta_i$ where variable $x_j$ appears not negated at position $i$ of $\Pi(\phi)$.
        \end{itemize}
        If instead $chall=1$, then the provers win if for each $i=1,...,m$, it is the case that $w_{3(i-1)+f_i}= a+\gamma_{i}$.
    \end{itemize}

    Recall Definition \ref{defin:proj}. The game $G^{3SAT}$ is $3^m$-projective where since after P1 has output values $w,w'$, then P2 has $3^m$ choices for $\Pi$ and $3^m$ choices for $f$, and in both cases the second value ($\delta$ or $c$) will be uniquely determined. This can be seen by rearranging the equations in the last bullet-point. In order to upper bound $\omega^*(G^{3SAT})$, we consider the game $G^{3SAT}_{coup}$, given in Definition \ref{defin:gamecoup}. From a winning strategy for $G^{3SAT}_{coup}$, we'll devise a strategy for P2 to guess $a$ based solely on P2's local input $chall$ and output.
    
    Fix an input/output pair $(a,(w_0,w_1))$ for P1, and suppose P2 successfully answers both challenges. Note that constructing a satisfying assignment is implicitly assumed to be impossible since soundness is currently being considered. Then there must be a variable $x_j$ and two conflicting clauses $i,i'\in\{1,...,m\}$ such that $x_j$ is negated at position $3(i-1)+f_i$ of $\Pi(\phi)$ but not negated at position $3(i'-1)+f_{i'}$ of $\Pi(\phi)$. Otherwise, a solution could be formed by following the red and green lines from the clauses to the variable assignments. This situation is illustrated below in Figure \ref{fig:3sat4}.

    \begin{figure}[ht!]
        \centering
        \includegraphics[scale=0.29]{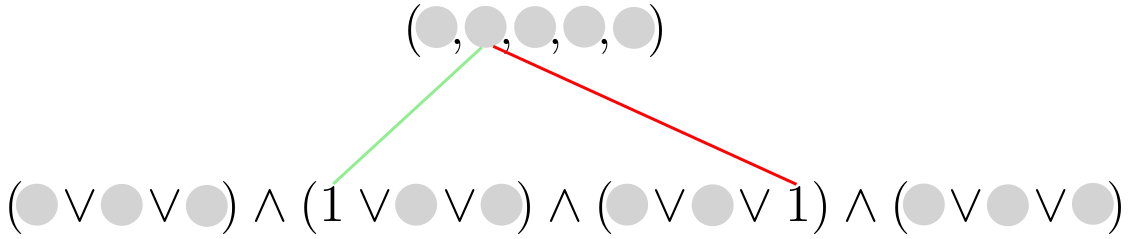}
        \caption{Conflicting variable assignments. The green line represents that the difference of the two bits was unveiled to be 0, and the red line means that the sum was unveiled to be 1\label{fig:3sat4}}
    \end{figure}
    
    Then the conflict between clauses $i,i'\in\{1,...,m\}$ for variable $x_j$ leads to the following.
    \begin{align*}
        w_{3(i-1)+f_i}+w'_j=a+\delta_{3(i-1)+f_i} \quad \text{ and }\quad w_{3(i'-1)+f_{i'}}-w'_j=\delta_{3(i'-1)+f_{i'}}
    \end{align*}
    
    Summing the contents of the two equations cancels out $w'_j$.
     \begin{align*}
        w_{3(i-1)+f_i}+w_{3(i'-1)+f_{i'}}=a+\delta_{3(i-1)+f_i}+\delta_{3(i'-1)+f_{i'}}
    \end{align*}
    
    Next, we use the equality from satisfying $chall=1$, and then finally isolate for $a$.
    \begin{align*}
        (a+\gamma_i)+(a+\gamma_{i'})=a+\delta_{3(i-1)+f_i}+\delta_{3(i'-1)+f_{i'}}\\
        \implies a=\delta_{3(i-1)+f_i}+\delta_{3(i'-1)+f_{i'}}-\gamma_{i}-\gamma_{i'}
    \end{align*}
    
    Note that the values on the right hand side are all decided by P2, hence P2 has a strategy to guess $a$. However, $P2$ can achieve this correlation with probability at most $\frac{1}{Q}$ by no-signalling, thus $\omega^*(G^{3SAT}_{coup})\leq \frac{1}{Q}$. Now, we can apply Proposition \ref{prop:bound1} from \cite{Chailloux_2017} to obtain    
    \begin{align*}
        \omega^*(G^{3SAT})\leq \frac{1}{2}+\left(\frac{64\cdot 3^m}{Q}\right)^{1/3}.
    \end{align*}
    
    If we take $Q\geq 64\cdot 2^{\log(3)m+3k}$, then the protocol has soundness $\frac{1}{2}+2^{-k}$. Next, there is a polynomial amount of communication in total since each step requires only a polynomial number of bits in $n$ and $m$ since $\log(Q)\in O(m)$. Indeed, in (1), V1 sends $\log(Q)$ bits, then in (2), P1 sends $(n+3m)\log(Q)$ bits, then in (3), only 1 bit is sent.  Finally in (4), if $chall=0$ then sending $\Pi$ takes at most $2m$ bits and sending the rest is $3mQ$ bits. On the other hand, if $chall=1$, then at most $m(Q+2)$ bits are sent. 
    \qed
\end{proof}

\begin{proposition}
    The 3SAT ZKP has perfect zero-knowledge against malicious quantum verifiers.
\end{proposition}

\begin{proof}
    We'll again use the model of 2-provers and a single quantum verifier, giving the verifier the freedom to query the provers in any order, without needing to respect relativistic timing restraints. As before, we model a cheating verifier $V^*$ as two families of circuits $(V_1^*,V_2^*)$, where $V_i^*$ takes as input a sub-view and outputs the message to prover $i$ for an instance of size $n$. Since the verifier can ask the questions in either order, we have two cases:
    
    \begin{itemize}
        \item Case 1: $V_2^*$ depends on the interaction with P1.
        \item Case 2: $V_1^*$ depends on the interaction with P2.
    \end{itemize}
    
    Since the cases are treated very similarly, we will just present the proof of Case 1. In other words, we'll consider a verifier that queries P1 and waits for a response before querying P2. Our first step is to describe the view when $V^*$ interacts with two honest provers. Then, we will define a simulator that can create the same view using only query access to $V^*$ despite not having any access to provers. Since these two views will be the same, then the protocol will be zero-knowledge.
    
    The view will consist of classical registers $Q_1$ and $Q_2$ that will hold the questions that will be asked to P1 and P2. Also, the responses of the two provers will be stored in classical registers $R_1$ and $R_2$. In addition, the verifier will hold a private quantum register $V$. In addition, we will adopt the notation $D(\psi):=|\psi\rangle\langle\psi|$ for quantum states to avoid the need to write $\psi$ twice.\\

    \textbf{Case 1 with Honest Prover}
     
    We assume that the operation of $V_2^*$ will depend on the interaction with P1.
    At the beginning of the protocol, the verifier's view is an auxiliary state $\sigma_0:=\rho_V$. 
    
    Next, after the verifier's first message, the view is
    \begin{align*}
        \sigma_1:=V_1^*(\rho_V)=\sum_{a\in \mathbb{F}_Q} p_{a}D(a)_{Q_1}\otimes \rho(a)_V.
    \end{align*}

    Here, $p_a$ is the probability of $a$ being the query for P1, and $\rho(a)_V$ is the verifier's private quantum state after sending $a$ to P1. Following the response of P1, the verifier's view is 
    \begin{align*}
        \sigma_2:=\frac{1}{Q^{n+3m}}\frac{1}{3^m}&\sum_{c'\in \mathbb{F}_Q^n,c\in \mathbb{F}_Q^{3m} }\sum_{\Pi\in CP}\\
        &\sum_{a\in \mathbb{F}_Q} p_{a} D(a\cdot s'+c',a\cdot p+c)_{R_1}\otimes D( a)_{Q_1}\otimes \rho(a)_V.
    \end{align*}
    
    Note that $p$ is uniquely determined by $\phi$, $\Pi$, and $s'$, so it is not part of the sum. Also, $|CP|=3^m$. Next, the verifier sends the challenge, which can be influenced by all that has happened up to this point. In other words, the verifier applies the circuit $V_2^*$ to the view $\sigma_2$. The view becomes
    
    \begin{align*}
        \sigma_3:=\frac{1}{Q^{n+3m}}\frac{1}{3^m} &\sum_{c'\in \mathbb{F}_Q^n,c\in \mathbb{F}_Q^{3m} }\sum_{\Pi\in CP}\sum_{a\in \mathbb{F}_Q} \sum_{chall\in\{0,1\}} p_{a,chall}  D(a\cdot s'+c',a\cdot p+c)_{R_1}\\
        &\otimes D(chall)_{Q_2} \otimes D(a)_{Q_1}\otimes \rho(a,chall,a\cdot s'+c',a\cdot p+c)_V.
    \end{align*}
    
    After the final message, we have two cases. If the challenge is 0, then P2 sends the CNF permutation $\Pi$ and a vector $\delta\in \mathbb{F}_Q^{3m}$ in the second response register. On the other hand, if the challenge is 1, then that register should instead contain vectors $\Pi(e)\in\{1,2,3\}^m$,$\gamma\in \mathbb{F}_Q^{m}$. In particular, an honest prover will choose $\delta_i=c_i\pm c_j'$ as described in step 4 of the protocol, and $\gamma_i=c_{3(i-1)+\Pi(e)_i}$ where $\Pi(e)_i$ indicates the position of a 1 in the $i$-th clause of $\Pi(\phi)$.
    
    \begin{align*}
        \sigma_4:=\frac{1}{Q^{n+3m}}\frac{1}{3^m} &\sum_{c'\in \mathbb{F}_Q^n,c\in \mathbb{F}_Q^{3m} }\sum_{\Pi\in CP}\sum_{a\in \mathbb{F}_Q} D(a\cdot s'+c',a\cdot p+c)_{R_1} \otimes D(a)_{Q_1}\otimes\\
        &\Big(p_{a,0}D(0)_{Q_2} \otimes D(\Pi,\delta)_{R_2} \otimes \rho(a,0,a\cdot s'+c',a\cdot p+c)_V\\
        &+p_{a,1}D(1)_{Q_2} \otimes D(\Pi(e),\gamma)_{R_2} \otimes \rho(a,1,a\cdot s'+c',a\cdot p+c)_V\Big)
    \end{align*}
    
    This is the final view when interacting with an honest prover. Now, before moving onto the simulator, we will rewrite this final state so that it will resemble the simulator's final state later on. The first step is to let $w':=a\cdot s'+c'$ and $w:=a\cdot p+c$ and sum over $w',w$ instead of $c',c$. Rewriting the sum this way may seem problematic since $\delta$ and $\gamma$ were explicitly defined using $c$ and $c'$, however using the relations $c'=w'-a\cdot s'$ and $c=w-a\cdot p$, the definition of $\gamma$ becomes
    
    \begin{align*}
        \gamma_i=c_{3(i-1)+\Pi(e)_i}=w_{3(i-1)+\Pi(e)_i}-a\cdot p_{3(i-1)+\Pi(e)_i}=w_{3(i-1)+\Pi(e)_i}-a.
    \end{align*}
    since the honest prover gives $\Pi(e)_i$ such that $p_{3(i-1)+\Pi(e)_i}=1$. The definition of $\delta$ becomes 
    \begin{align*}
        \delta_i =  
        \begin{cases} 
            w_i+w'_j-a & \text{if variable $x_j$ is negated at position $i$ of $\Pi(\phi)$.}\\
            w_i-w'_j & \text{if variable $x_j$ is not negated at position $i$ of $\Pi(\phi)$.}
        \end{cases}
    \end{align*}
    where $x_j$ is the variable at position $i$ of $\Pi(\phi)$. Now, we can replace the summation over $c',c$ with a summation over $w',w$ in $\sigma_4$ to obtain:
        
    \begin{align*}
        \sigma_4:=\frac{1}{Q^{n+3m}}\frac{1}{3^m} &\sum_{w'\in \mathbb{F}_Q^n,w\in \mathbb{F}_Q^{3m} }\sum_{\Pi\in CP}\sum_{a\in \mathbb{F}_Q} D(w',w)_{R_1} \otimes D(a)_{Q_1}\otimes\\
        &\Big(p_{a,0}D(0)_{Q_2} \otimes D(\Pi,\delta)_{R_2} \otimes \rho(a,0,w',w)_V\\
        &+p_{a,1}D(1)_{Q_2} \otimes D(\Pi(e),\gamma)_{R_2} \otimes \rho(a,1,w',w)_V\Big).
    \end{align*}
    
    The next step is to move the sum over $\Pi$ past the terms on which it does not act. We obtain
    \begin{align*}
        \sigma_4:=\frac{1}{Q^{n+3m}} &\sum_{w'\in \mathbb{F}_Q^n,w\in \mathbb{F}_Q^{3m} }\sum_{a\in \mathbb{F}_Q} D(w',w)_{R_1} \otimes D(a)_{Q_1}\otimes\\
        &\Big(p_{a,0}D(0)_{Q_2} \otimes \frac{1}{3^m}\sum_{\Pi\in CP} D(\Pi,\delta)_{R_2} \otimes \rho(a,0,w',w)_V\\
        &+p_{a,1}D(1)_{Q_2} \otimes  \frac{1}{3^m}\sum_{\Pi\in CP} D(\Pi(e),\gamma)_{R_2} \otimes \rho(a,1,w',w)_V\Big).
    \end{align*}
    At this point, we can rename $\Pi(e)$ to $f$ and sum over all $f\in\{1,2,3\}^m$ instead. We obtain 
    \begin{align*}
        \sigma_4:=\frac{1}{Q^{n+3m}} &\sum_{w'\in \mathbb{F}_Q^n,w\in \mathbb{F}_Q^{3m} }\sum_{a\in \mathbb{F}_Q} D(w',w)_{R_1} \otimes D(a)_{Q_1}\otimes\\
        &\Big(p_{a,0}D(0)_{Q_2} \otimes \frac{1}{3^m}\sum_{\Pi\in CP} D(\Pi,\delta)_{R_2} \otimes \rho(a,0,w',w)_V\\
        &+p_{a,1}D(1)_{Q_2} \otimes  \frac{1}{3^m}\sum_{f\in\{1,2,3\}^m} D(f,\gamma)_{R_2} \otimes \rho(a,1,w',w)_V\Big).
    \end{align*}
    
    \textbf{Case 1 with Simulator}
    
    Now, we describe how to simulate the views of the verifier without the help of any provers. We will denote the $i$th simulated view as $\tilde{\sigma}_i$. Since the simulator has access to $V_1^*$ and $\rho_V$, then $\sigma_0$ and $\sigma_1$ are straightforward to simulate. Intuitively, no effort is required at this stage because the prover has not acted yet. For $\sigma_2$, the response from P1 is two uniformly random vectors $w\in\mathbb{F}_Q^{3m},w'\in\mathbb{F}_Q^{n}$ since $c$ and $c'$ act as one-time pads. Then
    \begin{align*}
        \tilde{\sigma_2}=\frac{1}{Q^{n+3m}}\sum_{w\in\mathbb{F}_Q^{3m},w'\in\mathbb{F}_Q^{n}}\sum_{a\in \mathbb{F}_Q} p_{a}  D(w',w)_{R_1}\otimes      D(a)_{Q_1}\otimes \rho(a)_V.
    \end{align*}
    
    This can be created from $\tilde{\sigma_1}$ by tensoring with the maximally mixed state in register $R_1$. Next, the simulator applies $V_2^*$ to $\tilde{\sigma_2}$ to get $\tilde{\sigma_3}$:
    \begin{align*}
        \tilde{\sigma_3}=\frac{1}{Q^{n+3m}} &\sum_{w\in\mathbb{F}_Q^{3m},w'\in\mathbb{F}_Q^{n}}\sum_{a\in \mathbb{F}_Q} \sum_{chall\in\{0,1\}} p_{a,chall}  D(w',w)_{R_1}\\
        &\otimes D(chall)_{Q_2} \otimes D(a)_{Q_1}\otimes \rho(a,chall,w',w)_V.
    \end{align*}
    
    The final step is to simulate $\sigma_4$. Again, this depends on the challenge bit $chall$. 
    \begin{itemize}
        \item For $chall=0$ in register $Q_2$, the simulator chooses a random CNF permutation $\Pi$ and then puts $D(\Pi,\tilde{\delta})$ in register $R_2$ where $\tilde{\delta}\in\mathbb{F}_Q^{3m}$ is defined below
        \begin{align*}
            \tilde{\delta}_i =  \begin{cases} w_i+w'_j-a & \text{if variable $x_j$ is negated at position $i$ of $\Pi(\phi)$}\\
            w_i-w'_j & \text{if variable $x_j$ is not negated at position $i$ of $\Pi(\phi)$}
            \end{cases}
        \end{align*}
        where $x_j$ is the variable at position $i$ of $\Pi(\phi)$.
        \item For $chall=1$ in register $Q_2$, the simulator chooses a random $\tilde{f}\in \{1,2,3\}^m$ and then places $D(\tilde{f}, \tilde{\gamma})$ in register $R_2$, where $\tilde{\gamma}\in \mathbb{F}_Q^m$ is defined below:
        \begin{align*}
            \tilde{\gamma}_i= w_{3(i-1)+\tilde{f_i}}-a.
        \end{align*}
    \end{itemize}
    This gives a final view of 
    \begin{align*}
        \tilde{\sigma_4}=\frac{1}{Q^{n+3m}} &\sum_{w\in\mathbb{F}_Q^{3m},w'\in\mathbb{F}_Q^{n}}\sum_{a\in \mathbb{F}_Q}   D(w',w)_{R_1}\otimes D(a)_{Q_1}\otimes\\
        &\Big(p_{a,0}D(0)_{Q_2}\otimes \frac{1}{3^m}\sum_{\Pi\in PERM}D(\Pi,\tilde{\delta})_{R_2}\otimes\rho(a,0,w',w)_V\\
        &+p_{a,1}D(1)_{Q_2}\otimes \frac{1}{3^m}\sum_{\tilde{f}\in\{1,2,3\}^m}D(\tilde{f},\tilde{\gamma})_{R_2}\otimes\rho(a,1,w',w)_V\Big).
    \end{align*}
    Now, it is clear that the two views $\sigma_4$ and $\tilde{\sigma_4}$ are equal since $\tilde{\delta},\tilde{\gamma}$ are defined identically to $\delta,\gamma$ in the honest case, hence our simulator succeeded.
    
    The proof of case 2 follows the same logic.
    \qed
\end{proof}

\section{Conclusion and Future Work}

The zero-knowledge protocols presented in this paper highlight the usefulness of homomorphic bit commitment. Using the relativistic commitment scheme described in Section \ref{sec:REL-CS}, contents of commitments can be easily combined to create new, efficient protocols. In particular, the Subset Sum ZKP is efficient when compared to the existing quantum secure ZKPs in the literature. In addition, the sub-protocol that allows provers to prove that two commitments are equal is extremely simple with this approach by unveiling the difference of the contents. The proof technique used for soundness in \cite{Chailloux_2017} is extremely powerful and can be applied to a wide range of protocols. 

Note that in both protocols from Sections \ref{sec:SSP} and \ref{sec:3SAT}, the commitments that were combined were commitments made only by the provers. It would be interesting to devise a protocol that verifies whether two strings are equal when one string is in the possession of one party and the other string is kept by the other party. Indeed, this could allow for novel zero-trust identification protocols that have no computational assumptions, only relativistic ones. One difficulty of directly applying the existing approach to this problem is that at the end of the protocol, the parties should not necessarily know the difference of their two strings. They should only learn one bit of information from the protocol: whether the two strings are equal or not.

Finally, in \cite{https://doi.org/10.48550/arxiv.2112.01386}, the authors managed to prove a similar bound as in \cite{Chailloux_2017}, except that it can be applied to protocols with three challenges. If these bounds could be extended to prove soundness for protocols with a polynomial number of challenges, then perhaps a similar proof technique could be applied to prove soundness of the ZKP for QMA presented in \cite{Broadbent_2022} when using the relativistic commitment scheme seen in Section \ref{sec:REL-CS}. This would allow for a simple two-prover ZKP for QMA using only relativistic assumptions.

%
%
%
\bibliographystyle{alpha}
\bibliography{mybibliography}

\end{document}